\pdfoutput=1

\documentclass[10pt,conference,a4paper]{IEEEtran}

\usepackage{calc}
\usepackage{graphicx}
\usepackage{color}
\usepackage{epstopdf}
\usepackage[margin=0.5in]{geometry}
\usepackage{url}
\usepackage{amsfonts}
\usepackage{amssymb}
\usepackage{amsmath}
\usepackage{graphicx}
\usepackage{epsfig}
\usepackage{enumerate}
\usepackage{pgfplots} 
\pgfplotsset{compat=newest} 
\pgfplotsset{plot coordinates/math parser=false}

\newtheorem{lemmabody}{Lemma}
\newenvironment{lemma}{
	\begin{lemmabody}
	}{
\end{lemmabody} 
}
\newenvironment{proof}{
	{\it Proof:}
}{
$\Box$
}

\DeclareMathOperator{\tr}{tr}
\DeclareMathOperator{\Ev}{E}

\begin{document}
\title{Analog Coding of a Source with Erasures}

\author{\IEEEauthorblockN{Marina Haikin}
\IEEEauthorblockA{EE - Systems Department\\Tel Aviv University\\
Tel Aviv, Israel\\
Email: mkokotov@gmail.com}
\and
\IEEEauthorblockN{Ram Zamir}
\IEEEauthorblockA{EE - Systems Department\\Tel Aviv University\\
Tel Aviv, Israel\\
Email: zamir@eng.tau.ac.il}}


\maketitle
\begin{abstract}
Analog coding decouples the tasks of protecting against erasures and noise. For erasure correction, it creates an "analog redundancy" by means of band-limited discrete Fourier transform (DFT) interpolation, or more generally, by an over-complete expansion based on a frame. We examine the analog coding paradigm for the dual setup of a source with "erasure" side-information (SI) at the encoder. The excess rate of analog coding above the rate-distortion function (RDF) is associated with the energy of the inverse of submatrices of the frame, where each submatrix corresponds to a possible erasure pattern. We give a partial theoretical as well as numerical evidence that a variety of structured frames, in particular DFT frames with difference-set spectrum and more general equiangular tight frames (ETFs), with a common MANOVA limiting spectrum, minimize the excess rate over all possible frames. However, they do not achieve the RDF even in the limit as the dimension goes to infinity.
\end{abstract}


\section{Introduction} \label{Introduction}
Consider an i.i.d source sequence ${\textbf{x}=(x_1,..,x_n)}$ from a normal distribution $\mathcal{N}(0,\sigma^2_x)$. The encoder has information regarding the indices of $k$ important samples. Denote by ${\textbf{s}=(s_1,..,s_n), s_i\in\{0,1\}}$, the vector of this side information. The decoder must reconstruct an ${n}$-dimensional vector ${\bf \hat{x}}$ where only the values of samples dictated by ${\bf s}$ matter, while at the non-important samples the distortion is zero:
\begin{equation} \label{eqDistortion}
D(x,\hat{x},s) = 
\begin{cases}
(\hat{x}-x)^2,& \text{if } s= 1 \text{ (important) } \\
0,              & \text{if } s=0 \text{ (not important) }.\\
\end{cases}
\end{equation}
In \cite{SCdistortionSI} it is shown that in this setting of "erasure" distortion, when the SI is Bernoulli(${p}$) process, the encoder side information is sufficient, and the RDF is equal to that in the case where the side information is available to both the encoder and decoder: 
\begin{equation} \label{eqTheoreticalRD}
R(D) = \frac{p}{2}\log\bigg(\frac{\sigma^2_x}{D}\bigg).
\end{equation}
Here ${R}$ is the rate per source sample, ${D}$ is the average distortion at the important samples, and ${p=\frac{k}{n}}$ represents the probability of important samples. This rate can be achieved by a "digital" coding scheme; i.e. an "${n}$"-dimensional random code with joint-typicality encoding, at an {\it exponential} cost in complexity \cite{SCdistortionSI}.

In this paper we explore the following low complexity "interpolate and quantize" analog coding scheme:
\begin{equation} \label{IQscheme}
\textbf{T}_{enc} \rightarrow Q \rightarrow \textbf{T}_{dec}
\end{equation}
and its achievable rate for different transforms ${\textbf{T}_{enc}}$ and ${\textbf{T}_{dec}}$. Here, ${\textbf{T}_{enc}}$ is an ${n:m}$ linear transformation that depends on ${\bf s}$, ${\textbf{T}_{dec}}$ is an ${m:n}$ linear transformation that is independent of ${\bf s}$, for some ${n\ge m\ge k}$, and ${Q(\cdot{})}$ denotes quantization. Typically we consider a constant ratio of important samples i.e ${k\approx n/2}$ and are interested in the asymptotic performance (${n\rightarrow \infty}$).

One motivation for the analog coding scheme comes from the solution given in \cite{SCdistortionSI} to a lossless version of this problem. In this setting, the encoder uses the Reed Solomon (RS) decoding algorithm to correct the erasures and determine the ${k}$ information symbols. It then transmits these symbols to the decoder at a rate of  ${\frac{k}{n}\log(J)}$ bits per sample, where ${J}$ is the size of the source alphabet. To reconstruct the source, the decoder uses the RS encoding algorithm to get the ${n}$ reconstructed samples, that coincide with the source at the ${k}$ non-erased samples, as desired. we can view the RS decoder as a system which performs {\it interpolation} of the erased source signal.

Such an approach could be extended to a continuous source, if we first quantize the important samples to  ${J}$ levels and then apply the RS code solution above.  However, this "quantize and interpolate" solution is limited to scalar quantization. In contrast, the scheme in (\ref{IQscheme}) reverses the order of quantization and interpolation and therefore it is not limited to scalar quantization. However, the interpolation step (${\textbf{T}_{enc}}$) typically suffers from a {\it signal amplification} phenomenon. This is the main issue we deal with as it results in an increase in rate. 
 
Our problem formulation is dual to Wolf's paradigm of analog channel coding, in which transform techniques are exploited for coding in the presence of impulse noise \cite{Wolf}. Wolf's scheme decouples impulse correction - by analog means - and additive white Gaussian noise (AWGN) protection - by digital means. The impulse-pattern dependent transform at the decoder introduces {\it noise amplification} for a general impulse pattern. In our case, the digital component is the quantizer, which is responsible for the lossy coding. The transform at the encoder  causes signal amplification whose severeness depends on the pattern of important samples. 

The main question which we explore is whether analog coding can asymptotically achieve the optimum information-theoretic solution (\ref{eqTheoreticalRD}). And even if not, what are the best tranforms ${\textbf{T}_{enc}}$ and ${\textbf{T}_{dec}}$ in (\ref{IQscheme}). Our preliminary results are unfortunately negative: the coding rate of the scheme in (\ref{IQscheme}) is strictly above the RDF, even for the best transforms, and even if we let the dimension $n$ go to infinity. 

${\textbf{T}_{dec}}$ can be considered as a frame used for dimension expansion after the dimension reduction performed at the encoder.
Several works explored frames which are good for other applications. In compressed sensing, for example, most commonly the goal is to maximize the spectral norm for all sub-matrices \cite{Candes}. In \cite{FrameTheoreticDFT}, \cite{systematicDFTframes} frames for coding with erasures are introduced  but they are tolerant only to specific patterns. In \cite{OptimumFramesErasure} ETFs are analyzed but only for small amount of erasures.

The main contributions introduced in this paper are the asymptotic point of view - concentration properties and universality of different structured frames, and how they compare to random i.i.d transforms; the redundant sampling (${m>k}$) approach; and the empiric observation that some ETFs are optimum (or at least local minimizers) in the sense of average signal amplification over all erasure patterns.
Section~\ref{SystemCharacterization} describes the analog coding scheme in detail. Section~\ref{TransformOptimization} analyzes the performance of a random i.i.d transform which turns out to be better than a low pass DFT frame, while Section~\ref{IrregularSpectrum} explores the superior approach of difference-set spectrum, random spectrum and general ETFs. 
\section{System Characterization }\label{SystemCharacterization}
We begin with defining the system model. 
A {\it Transform Code} is characterized by a "universal" transform at the decoder and pattern dependent transform at the encoder. 
Let ${\bf A}$ be the ${n\times m}$ matrix representation of a frame with ${n}$ ${m}$-dimensional elements as rows, where ${\frac{m}{k}\triangleq\beta}$ is the redundant-sampling ratio, which varies in the range ${1\le\beta\le\frac{1}{p}}$.\footnote{This is unlike the conventions in frame theory in which the frame elements are column vectors.} ${\bf A}$ will serve as a transformation applied at the decoder, (${\textbf{T}_{dec}}$ in (\ref{IQscheme})), independent of the pattern of important samples, as the side information is not available to the decoder. The pattern ${\bf s}$ of the important samples defines which ${k}$ rows of ${\bf A}$ contribute to meaningful values.
The corresponding rows construct the ${k\times m}$ transform ${\textbf{A}_s}$.
Denote by ${\textbf{B}_s}$ the ${m\times k}$ transform applied at the encoder, (nonzero part of ${\textbf{T}_{enc}}$ in (\ref{IQscheme})), to the vector ${\textbf{x}_s}$ of important samples. ${\textbf{f}=\textbf{B}_s\textbf{x}_s}$ is the vector of transformed samples. As illustrated in Figure~\ref{modelFig}, we use a white additive-noise model for the quantization of the ${m}$ transformed samples, e.g. entropy coded dithered quantization (ECDQ), \cite{AdamNoise}, which is blind to the locations of the important samples. As we recall, in the no-erasure case this additive-noise model can achieve the RDF with Gaussian noise (corresponding to large dimensional quantization) and Wiener filtering \cite{Berger}. 
\begin{figure}[htbp]
	\hspace{-1.5em}
	\def\svgscale{2}
	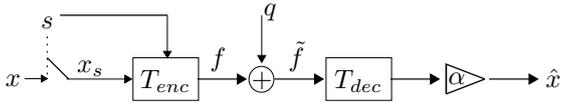
	\caption{Analog coding scheme. ${x,\hat{x}}$ are ${n}$-dimensional vectors, ${x_s}$ is ${k}$-dimensional and ${f,\tilde{f}}$ and ${q}$ are ${m}$-dimensional.}
	\label{modelFig}
\end{figure}\\
The ${k}$ reconstructed important samples are thus
\begin{equation} \label{Model}
{\bf \hat{x}}_s=\alpha\textbf{A}_s{\bf\tilde{f}}=\alpha\textbf{A}_s(\textbf{f}+\textbf{q})=\alpha\textbf{A}_s\textbf{B}_s\textbf{x}_s+\alpha\textbf{A}_s\textbf{q}
\end{equation}
where ${\bf\tilde{f}}$ is the quantized version of the transformed samples and ${\bf q}$ is a white quantization noise with variance ${\sigma^2_{q}}$.
We can deal separately with the choice of ${\textbf{B}_s}$ and ${\alpha}$. The encoder applies ${\textbf{B}_s}$ such that ${\textbf{A}_s\textbf{B}_s=\textbf{I}}$ and ${\alpha}$ is a Wiener coefficient.

Let ${\|\textbf{A}_s\|^2}$ denote the squared Frobenius norm of the matrix ${\textbf{A}_s}$ normalized by the number of rows ${{\|\textbf{A}_s\|^2}=\frac{1}{k} \|\textbf{A}_s\|_F^2=\frac{1}{k}\sum_{i=1}^{k}\|{\textbf{A}_s}_i\|^2}$, where ${\|{\textbf{A}_s}_i\|}$ is the ${\it l_2}$ norm of the ${i}$'th row. 
\subsection{\bf Rate - Distortion Derivation}
Since the decoder is blind to the transform, the rate is given by that of a white input with the same average variance \cite{Lapidot}. The rate per sample for a pattern ${\textbf{s}}$ is therefore the mutual information in an AWGN channel with a white Gaussian input: 
\begin{equation} \label{RateAllRes1}
R=\frac{m}{n}\frac{1}{2}\log\bigg(1+\frac{\frac{1}{m}E\|\textbf{f}\|^2}{\sigma^2_{q}}\bigg)
\end{equation}
\begin{equation} \label{RateAllRes2}
=\frac{m}{n}\frac{1}{2}\log\bigg(1+\frac{\sigma^2_x}{\sigma^2_{q}}\|\textbf{B}_s\|^2\bigg)
\end{equation}
where to obtain (\ref{RateAllRes2}) we substitute the average variance of the transformed samples:
\begin{equation} \label{VarFredund}
\frac{1}{m}E\|\textbf{f}\|^2=\frac{1}{m}\sum_{i=1}^{m}\sigma^2_{f_i}=\frac{1}{m}\sum_{i=1}^{m}\|{{\textbf{B}}_s}_i\|^2\sigma^2_x=\|\textbf{B}_s\|^2\sigma^2_x.
\end{equation}
For a given ${\textbf{A}_s}$, the matrix ${\textbf{B}_s}$ that minimizes the expected ${\it l_2}$ norm of ${\bf f}$ in (\ref{RateAllRes1}) is the pseudo-inverse  ${\textbf{B}_s=\textbf{A}_s'(\textbf{A}_s\textbf{A}_s')^{-1}}$, hence\footnote{${(\ )'}$ denotes the conjugate transpose.} 
\begin{equation} \label{PointScore}
\|\textbf{B}_s\|^2=\frac{1}{m}\|\textbf{B}_s\|_F^2=\frac{1}{m}\tr(\textbf{B}_s'\textbf{B}_s)=\frac{1}{m}\tr((\textbf{A}_s\textbf{A}_s')^{-1}).
\end{equation}

We shall later see that the heart of the problem is the signal amplification caused by the factor ${\|\textbf{B}_s\|^2}$ in (\ref{RateAllRes2}).
The case of ${m>k}$ is referred to as "redundant sampling", where more samples are quantized than the important ones. The motivation for this is the existence of more robust transforms in the sense of signal amplification even at the cost of some extra transmissions.

For convenience, we normalize the transform ${\bf A}$ to have unit-norm rows,
${\|{\textbf{A}}_i\|=1}$, for ${i=1,...,n}$, so that each sample of the additive quantization noise term in (\ref{Model}) has variance ${{\|\textbf{A}_s}_i\|^2\sigma^2_{q}=\sigma^2_{q}}$. 
The variance of each sample of ${\textbf{A}_s\textbf{B}_s\textbf{x}_s}$ is ${\sigma^2_x}$. As a result of the Wiener estimation the distortion is:
\begin{equation} \label{D}
D\triangleq E\bigg\{\frac{1}{k}\sum_{i=1}^{n}D(x_i,\hat{x}_i,s_i) \bigg\}=\frac{1}{k}\Ev\|\bf \hat{x}_s-\textbf{x}_s\|^2=\frac{\sigma^2_x\sigma^2_{q}}{\sigma^2_x+\sigma^2_{q}}.
\end{equation}
Combining (\ref{RateAllRes2}),(\ref{PointScore}) and (\ref{D}) we can relate the rate and distortion of the scheme for a specific pattern ${\bf s}$:
\begin{equation} \label{Rate7}
R=\frac{m}{n}\frac{1}{2}\log\bigg(1+\frac{\frac{1}{m}\tr((\textbf{A}_s\textbf{A}_s')^{-1})(\sigma^2_x-D)}{D}\bigg)
\end{equation}

We define the excess rate of the scheme as ${\delta \triangleq R-R(D)}$:
\begin{equation} \label{ExcessRate2}
\delta(\beta,\gamma,\eta_s)=\frac{k}{n}\frac{1}{2}\big[\beta\log(\eta_s\gamma+(1-\eta_s))-\log(\gamma)\big]
\end{equation}
\begin{equation} \label{ExcessRate3}
\simeq\beta \frac{p}{2}\log(\eta_s)+(\beta-1)\frac{p}{2}\log(\gamma)
\end{equation}
where ${\gamma=\frac{\sigma^2_x}{D}}$ is the signal-to-distortion ratio (SDR), 
\begin{equation} \label{IE}
\eta_s=\frac{1}{m}\tr((\textbf{A}_s\textbf{A}_s')^{-1})
\end{equation}
is the inverse energy (IE) of a pattern ${\textbf{s}}$, which is related to harmonic mean of the eigenvalues of ${\textbf{A}_s\textbf{A}_s'}$, and ${\simeq}$ is true for high resolution (${\gamma\gg1}$).
We also define ${\rho}$ as the mean logarithmic inverse energy (MLIE) of the frame ${\bf A}$:
\begin{equation} \label{MLIE}
\rho=\frac{1}{{n \choose k}}\sum_{s}\frac{m}{n}\frac{1}{2}\log(\eta_s)
\end{equation}
i.e. the average (over all possible patterns of ${k}$ important samples) excess rate above the RDF caused by signal amplification.
As we will see, for "good" transforms this average becomes asymptotically the typical case.

In the high resolution case, the excess rate (\ref{ExcessRate3}) is composed of two terms, one as the result of signal amplification and the second as a result of ${m-k}$ extra samples transmission. 
Note that for fixed (${n,k,\beta}$) minimizing the excess rate ${\delta}$  is equivalent to minimizing the MLIE ${\rho}$.

\subsection{\bf Side Information Transmission}\label{SI}
It is most natural to compare the proposed system with the alternative  naive approach of transmitting the side information regarding the locations of the important samples. Pattern transmission requires ${\frac{1}{n}\log({n \choose k})}$ bits per input sample, which is ${H_b(p)}$ bits in the limit ${n\to \infty}$.

\section{Transform Optimization} 
\label{TransformOptimization}
\subsection{\bf Band Limited Interpolation }
\label{BLinterpolation}
The most basic frame includes ${m}$ consecutive rows of the IDFT. Without loss of generality, the transform matrix ${\bf A}$ consists of the first ${m}$ columns of an IDFT matrix, meaning that the reconstructed samples are part of a band limited (lowpass) temporal signal with the quantized DFT coefficients as its spectrum (${m}$ lower frequencies). Such a transform ${\bf A}$ forms a "full spark frame" - every subset of ${m}$ rows of ${\bf A}$ is independent, i.e ${\textbf{A}_s\textbf{A}_s'}$ is full rank and invertible for every pattern ${\bf s}$ \cite{FullSparkFrames}. However, it is not robust enough to different patterns. Intuitively, though the source samples are i.i.d, the band limited model forces slow changes between close samples. Thus, it is good for a uniform sampling pattern, but for most other patterns it suffers from signal amplification that causes a severe loss in rate-distortion performance \cite{AdamNoise}. Asymptotically almost every sub-matrix ${\textbf{A}_s}$ is ill-conditioned and the IE ${\eta_s}$ (\ref{IE}) is unbounded even for redundant sampling (${\beta>1}$).
Figure~\ref{figBL_IE} shows the IE distribution for band-limited interpolation
and a random pattern ${\bf s}$. The dashed line at ${\frac{1}{\beta-1}}$ is the asymptotic theoretic value achieved by a random i.i.d. transform;
see Section~\ref{RandomTransform}. 
\begin{figure}[ht]
	\begin{center}
		\includegraphics[scale=.49]{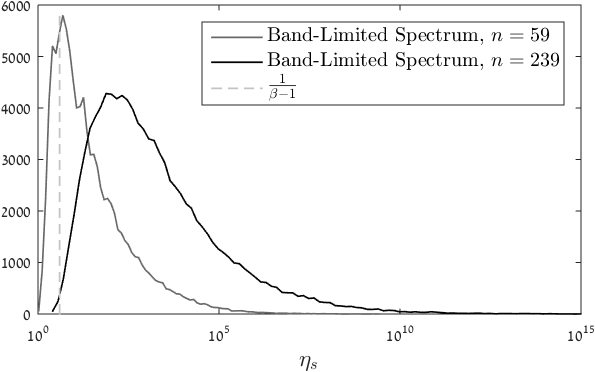}
		\caption{Logarithmic histogram of ${\eta_s}$, ${m=\lfloor\frac{n}{2}\rfloor}$, ${\beta=1.25}$.}
		\label{figBL_IE}
	\end{center}
\end{figure}
We can see that the IE diverges as ${n}$ grows.
\subsection{\bf Signal Amplification} \label{Signal Amplification}  
The following "Inversion-Amplification Lemma" describes the condition for an optimal transform.
\begin{lemma} \label{lemma1}
	{\em The IE ${\eta_s}$ in (\ref{IE}) of any ${k\times m}$ matrix ${\textbf{A}_s}$, s.t. ${\|{\textbf{A}_s}_i\|=1}$, is lower bounded as ${\eta_s\ge \frac{k}{m}}$, with equality iff ${\textbf{A}_s\textbf{A}'_s= \textbf{I}}$.}
\end{lemma}

\begin{proof}
	Denote by ${\{\lambda_i\}_{i=1}^{k}}$ the eigenvalues of ${\textbf{A}_s\textbf{A}'_s}$.
	\begin{align}
		&1=\frac{1}{k}\tr(\textbf{A}_s\textbf{A}'_s)=\frac{1}{k}\sum_{i=1}^{k}\lambda_{i}\ge \frac{1}{\frac{1}{k}\sum_{i=1}^{k}\frac{1}{\lambda_{i}}}=\frac{1}{\frac{1}{k}\tr((\textbf{A}_s\textbf{A}'_s)^{-1})}
		\nonumber \\
		& \Rightarrow \;
		\frac{1}{m}\tr((\textbf{A}_s\textbf{A}'_s)^{-1})=\frac{1}{\beta}\frac{1}{k}\tr((\textbf{A}_s\textbf{A}'_s)^{-1})\ge\frac{1}{\beta}.	
		\nonumber
	\end{align}
	where the inequality follows from the arithmetic-harmonic mean inequality, with equality iff all the eigenvalues are equal, i.e ${\textbf{A}_s\textbf{A}'_s= \textbf{I}}$.
\end{proof}

For ${\beta=1}$, Lemma~\ref{lemma1} becomes ${\|\textbf{A}^{-1}_s\|\ge1}$, with equality iff ${\textbf{A}_s}$ is unitary. Thus for a non-unitary transform the signal is amplified by the factor ${\|\textbf{A}^{-1}_s\|}$.

\subsection{\bf Random i.i.d Transforms } \label{RandomTransform}   
For "digital" coding, we know that random i.i.d codes are optimal. Thus, a natural approach is to investigate the asymptotic performance of a random transform. Consider a matrix ${\textbf A}$ whose entries are i.i.d Gaussian random variables with variance ${\frac{1}{m}}$. For any ${k\times m}$ sub-matrix ${{\textbf{A}_s}}$, ${\lim_{k\to \infty}\|{\textbf{A}_s}_i\|=1}$ almost surly. 

\subsubsection{\bf Amplification Analysis}
We bring here two results which show that for ${m=k}$, random i.i.d transform is definitely bad in the sense of amplification:

\begin{equation} \label{RandomAmp1}
\lim_{k\to \infty}P[\frac{1}{k}\tr((\textbf{A}_s\textbf{A}')^{-1}_s)\ge 1+\zeta]=1,\ \ \ \forall \zeta\ge0
\end{equation}
This means that with square random matrix we cannot achieve a 'non-amplifying' transformation. Moreover, w.p.1 the amplification diverges for large dimensions.
For  ${k \to \infty}$ we can bound the divergence rate as follows:
\begin{equation} \label{RandomAmpBound}
\frac{k^2}{2\pi e} \le E\bigg[\frac{1}{k}\tr((\textbf{A}_s\textbf{A}')^{-1}_s)\bigg]\le \frac{k^3}{2\pi e}
\end{equation}
The proof of both of these results is based on characterization of the minimum eigenvalue of a random matrix \cite{RandomMatrix1}, and is omitted due to space constraints.

For ${m>k}$ the amplification is finite. As random matrix theory shows, \cite{RandomMatrix1}, if ${\textbf{H}}$ is an ${r\times t}$ random matrix with i.i.d entries of variance ${\frac{1}{r}}$ and ${\frac{t}{r}\to \beta, \beta >1}$, then 

\begin{equation} \label{RandomMatrix2}
\lim_{r\to \infty}\frac{1}{r}\tr((\textbf{H}\textbf{H}')^{-1})=\frac{1}{\beta -1} \; \; a.s.
\end{equation}
A ${k\times m}$ sub-matrix ${\textbf{A}_s}$ has element variance of ${\frac{1}{m}}$. Denote ${\textbf{H}=\sqrt{\frac{m}{k}}\textbf{A}_s}$, which has element variance of ${\frac{1}{k}}$: 
\begin{displaymath}
\frac{1}{m}\tr((\textbf{A}_s\textbf{A}_s')^{-1})=\frac{1}{k}\tr((\textbf{H}\textbf{H}')^{-1})
\end{displaymath}
\begin{equation} \label{RandomAmp3}
\Rightarrow \lim_{k\to \infty}\eta_s=\frac{1}{\beta -1}.
\end{equation}
\subsubsection{\bf Comparison to the SI Transmission Benchmark (Section~\ref{SI})}
Substituting (\ref{RandomAmp3}) as the IE in (\ref{ExcessRate2}) we get the following expression for the excess rate using a random transform: 
\begin{equation} \label{RateLoss}
\delta=\frac{k}{n}\frac{1}{2}\bigg[\beta\log\bigg(\frac{1}{\beta-1}\gamma+1-\frac{1}{\beta-1}\bigg)-\log(\gamma)\bigg]
\end{equation}
For some scenarios this outperforms the naive side information transmission. 

Figure~\ref{figRateLossVsSDR} shows the asymptotic excess rate above (\ref{eqTheoreticalRD}) for random transform with optimal ${\beta}$ for each SDR compared to the cost of SI transmission.
\begin{figure}[ht]
	\begin{center}
		\includegraphics[scale=.49]{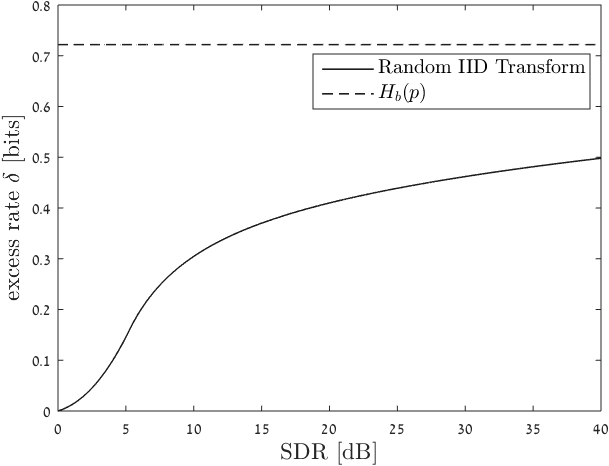}
		\caption{Rate loss for ${\frac{k}{n}=\frac{1}{5}}$.}
		\label{figRateLossVsSDR}
	\end{center}
\end{figure}
In the limit of high SDR the expression for the excess rate (for the best choice of ${\beta}$) takes the following form:
\begin{equation} \label{RateLossHighSDR}
\delta=\frac{k}{n}\frac{1}{2}\log(\ln(\gamma)),
\end{equation}
which goes (very slowly) to ${\infty}$. Nevertheless, for reasonably high SDR there is an advantage to the random matrix approach relative to the benchmark.
Analysis of ${\beta}$ which minimizes the rate loss in (\ref{RateLoss}) and the proof of (\ref{RateLossHighSDR}) are omitted here due to space limitations. 
\section{Irregular Spectrum } \label{IrregularSpectrum}     
As we saw in Section~\ref{TransformOptimization}, a band-limited DFT frame suffers from high signal amplification for non-uniform erasure patterns. In contrast, the signal amplification of a random frame is invariant to the erasure pattern. We can increase the robustness of a DFT-based frame to the erasure pattern by selecting an irregular "symmetry breaking" set of frequencies.
Thus, the encoder performs interpolation to a signal with irregular spectrum. 

Every choice of frequency pattern forms a frame ${\bf A}$, which consists of the columns of the IDFT matrix that correspond to the frequency pattern. 
For a general spectral patten the corresponding frame is not necessarily full spark (for a general ${n}$). But for prime ${n}$, Chebotarev's theorem guarantees that for every spectral choice and every pattern of important samples, ${\textbf{A}_s^{-1}}$ (or the pseudo inverse) exists \cite{Chebotarev}. We thus restrict the discussion to prime ${n}$'s when exploring the DFT transform.  

\subsection{\bf Difference-Set Spectrum} \label{DSS}  
For small dimensions it is possible to exhaustively check all spectrum choices and look for the one with the best worst case or average amplification (logarithmic IE). It turns out that the best spectrum consists of frequencies from a {\it difference set} (DS), forming the so called difference-set spectrum (DSS). 

{\it Definition:} an ${m}$ subset of ${\mathbb{Z}_n}$ is a ${(n,m,\lambda)}$ difference set if the distances (modulo ${n}$) of all distinct pairs of elements take all possible values ${1,2,..,n-1}$ exactly ${\lambda}$ times. The three parameters must satisfy the relation
\begin{equation}\label{DSrelation}
\lambda (n-1)=m(m-1). 
\end{equation}
Difference sets are known to exist for some pairs of ${(n,m)}$. We consider the case of prime ${n}$ and ${m\approx \frac{n}{2}}$. Specifically, we consider a {\it Quadratic Difference Set} \cite{WB_DSS} with the following parameters:
\begin{equation}\label{QuadraticDS}
\bigg(n=p, m=\frac{p-1}{2}, \lambda=\frac{p-3}{4}\bigg), \;\;\;\; p-prime.
\end{equation}
This DS can be constructed by a cyclic sub group ${\langle g\rangle}$, for some element ${g}$ from the multiplicative group of ${\mathbb{Z}_n}$. 

An ${n\times m}$ DSS transform ${\textbf{A}}$ is constructed from the ${m}$ columns of an ${n\times n}$ IDFT matrix, that correspond to indices from a difference set. The normalization of the IDFT is such that ${\|\textbf{A}_i\|=1}$, i.e the absolute value of the elements is ${\frac{1}{\sqrt{m}}}$. 

\subsection{\bf Random Spectrum and the MANOVA Distribution} \label{RandomSpectrum}  

Interestingly, asymptotically in this setup, a {\it random spectrum} achieves similar performance as the DSS. Recall that the IE is determined by the eigenvalue distribution (Lemma~\ref{lemma1}). Farrell showed in \cite{Farrell} that for a random spectrum, the limiting empirical eigenvalue distribution of ${\textbf{A}_s\textbf{A}_s'}$, for a randomly chosen ${\textbf s}$, converges almost surly to the limiting density of the Jacobi ensemble. The Jacobi ensemble corresponds to the eigenvalue distribution of MANOVA matrices - random matrices from multi-variate distribution \cite{MANOVA}. With our notations the MANOVA distribution is equal to:
\begin{equation}\label{Jacobi}
f^{M}(x)=\frac{\beta\sqrt{(x-r_-)(r_+-x)}}{2\pi x(1-\frac{m}{n}x)}\cdot I_{(r_-,r_+)}(x),
\end{equation}
\begin{equation}\label{JacobiExtrimalValues}
r_\pm=\bigg(\sqrt{(1-\frac{m}{n})\frac{1}{\beta}}\pm\sqrt{1-\frac{m}{n\beta}}\bigg)^2.
\end{equation}

We explore cases with ${\frac{m}{n}\rightarrow\frac{1}{2}}$, thus:
\begin{equation}\label{ManovaScore}
\lim_{n\to \infty}\eta_s=\frac{1}{\beta}E_{f^M}(X^{-1})=\frac{1}{\beta}\int\frac{1}{x}f^M(x)dx
\end{equation}
\begin{equation}\label{ManovaScore2}
=\int_{1-c}^{1+c}\frac{\sqrt{c^2-(x-1)^2}}{\pi x^2(2-x)}dx, \;\;\; c=\sqrt{\frac{1}{\beta}\bigg(2-\frac{1}{\beta}\bigg)}.
\end{equation}

Figure~\ref{figIE_rate} (on the left) shows the histogram of ${\eta_s}$ over randomly chosen patterns for a large dimension in the case of a random i.i.d transform and a DSS transform, as well as random spectrum transform.   
We see that the IE of a random i.i.d transform concentrates on ${\frac{1}{\beta-1}}$ (\ref{RandomAmp3}). For DSS/random spectra, almost all patterns (sub-matrices) are equivalent, and the IE concentrates on a lower value which fits the MANOVA density based calculation (\ref{ManovaScore2}). The ideal lower bound (of Lemma~\ref{lemma1}) is also presented.
The advantage of these structured transforms lies in their better eigenvalue distribution.

Figure~\ref{figEigenValues} shows the empirical distribution of the eigenvalues of ${\textbf{A}_s\textbf{A}_s'}$ for different transforms. Figure~\ref{figEigenValues} shows also the theoretical limiting eigenvalue density of an i.i.d random matrix (Marchenko-Pastur) \cite{RandomMatrix1} and the MANOVA distribution (\ref{Jacobi}).
\begin{figure}[ht]
	\begin{center}
		\includegraphics[scale=.49]{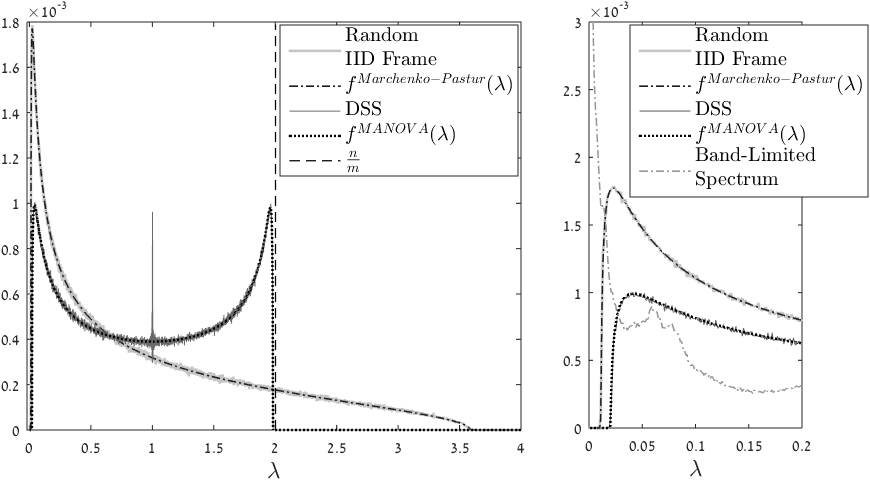}
		\caption{Eeigenvalue distribution of ${\textbf{A}_s\textbf{A}_s'}$, ${n = 947, \beta = 1.25}$.
			The graph on the right hand side zooms in into the behavior near zero.}
		\label{figEigenValues}
	\end{center}
\end{figure}
Observe the concentration property of the eigenvalue empirical distribution and of the IE ${\eta_s}$ (\ref{IE}); namely with high probability these random functionals are close to their mean value. It is evident that DSS also fits the asymptotic MANOVA density of a random spectrum.
Observe also that the minimal eigenvalue of a random i.i.d transform is closer to zero than that of DSS and thus its contribution to the IE amplification is more significant (see the zoom in graph on the right). As ${\beta}$ decreases, the extremal eigenvalues move towards the edges (0 and ${\frac{n}{m}}$), and the minimal eigenvalue becomes the most dominant for the IE. For ${\beta =1}$, the support of the density function approaches zero, and as a result the IE diverges
and there is no concentration to a finite value. Note that in band limited spectrum this is the case even for ${\beta > 1}$. 

\subsection{\bf Equiangular Tight Frames} \label{ETFs}  

It turns out that a DSS spectrum is a special case of an {\it equiengular tight frame} (ETF) or a {\it maximum-Welch-boud-equality codebook} (MWBE) {\cite{WB_DSS}}. Moreover, we observe that asymptotically many different ETFs are similar in terms of their ${\eta_s}$ distribution. 

An ${n\times m}$ ETF transform ${\textbf{A}}$ is defined as a tight frame (i.e., it satisfies ${\textbf{A}'\textbf{A}=\frac{n}{m}\textbf{I}_m}$) such that the absolute value of the correlation between two rows is constant for all pairs and equal to the Welch bound:
\begin{equation}\label{WBwc3}
|\textbf{a}_l\textbf{a}'_{l'}| = \sqrt{\frac{n-m}{(n-1)m}}=cos(\theta),\;\;\;\;\forall l\not= l'.
\end{equation}
The matrix ${\textbf{A}\textbf{A}'}$ is Hermitian positive semidefinite whose diagonal elements are 1 and whose off-diagonal elements have equal absolute value ${cos(\theta)}$ as in (\ref{WBwc3}). It has ${m}$ eigenvalues equal to ${\frac{n}{m}}$ (same as in ${\textbf{A}'\textbf{A}}$) and the rest ${n-m}$ eigenvalues are zero.
For any ${k\le m}$ rows of ${\textbf{A}}$ (induced by the important samples pattern ${s}$) ${\textbf{A}_s\textbf{A}_s'}$ is positive definite. The absolute value of the off-diagonal elements of the ${k\times k}$ matrix ${\textbf{A}_s\textbf{A}_s'}$ is also ${cos(\theta)}$ but its eigenvalue spectrum is induced by the subset of the element's phases. The distribution of this spectrum is the main issue of interest when exploring the IE ${\eta_s}$ (Figure~\ref{figEigenValues}).

While DFT-based transforms assume a complex source, ETFs allow us to consider also real valued sources, which is our original motivation. As stated before, other types of ETFs mentioned above achieve similar IE distribution.  
Figure~\ref{figIE_rate} (on the right) shows the histogram of ${\delta}$, as defined in (\ref{ExcessRate3}), for a real random i.i.d transform as well as Paley's real ETF \cite{Palye34inGrassmannian}, in high SDR. 
\begin{figure}[ht]
	\begin{center}
		\includegraphics[scale=.49]{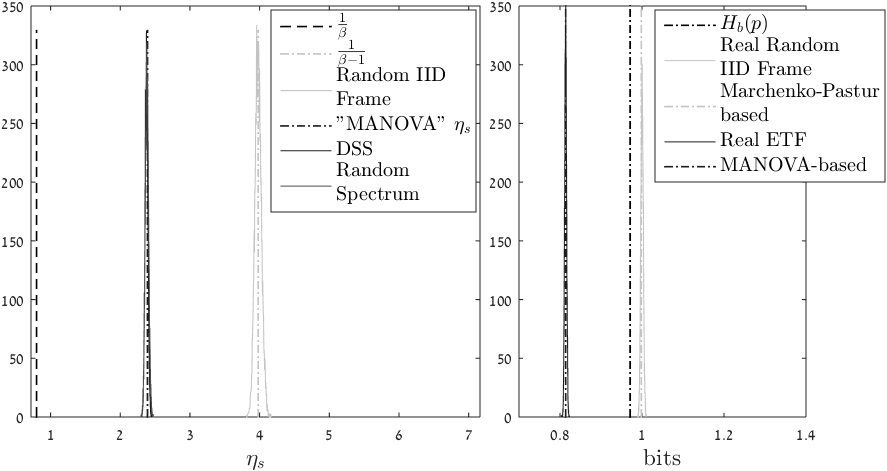}
		\caption{Left: Histogram of the inverse energy ${\eta_s}$ for ${n = 947, \beta = 1.25}$. 
			Right: Histogram of ${\delta}$, ${n = 1902, m = \frac{n}{2}, \beta = 1.25, SDR = 30dB}$.}
		\label{figIE_rate}
	\end{center}
\end{figure}
We can see that for this setup the rate of a random i.i.d transform exceeds that of SI transmission, but a scheme based on a real-valued ETF achieves a lower rate.

Finally, for given ${\frac{k}{n}}$ and ${\frac{m}{n}}$ values, we observe that many different frames are asymptotically equivalent and share a similar MANOVA eigenvalue distribution. In [16] we further study this asymptotic behavior for a gallery of structured frames, including different ETFs but not only. (As far as we know, no results were proved on the asymptotic spectra of submatrices of any frame other than a random choice of columns from DFT or Haar matrices.) 
We conjecture that these frames are asymptotically optimal for analog coding of a source with erasures. Moreover, we have a strong evidence that for every dimension where an ETF exists, it is optimal in the sense of the average excess rate caused by signal amplification; i.e, it minimizes the MLIE (\ref{MLIE}) over all possible ${(n,m)}$ frames. In particular, we verified for specific dimensions that DSS and real ETF (based on Paley’s construction of symmetric conference matrices) are local minimizers of the MLIE.
\section{Concluding Remarks } \label{Conc}   
Our results indicate that analog coding with "good" frames outperforms both band-limited interpolation and random i.i.d frames and beat the naive benchmark (Section~\ref{SI}), but it is inferior to the pure digital (rate-distortion function achieving) solution. DSS frames, and more generally - ETFs - seem to be the natural candidates for optimal analog-coding frames. Furthermore, a large family of frames, including these ones, seem to exhibit a common limiting behavior of a MANOVA spectrum. Although the question of whether they are the best frames remains open, our current results strongly support a positive answer. 

\section*{Acknowledgement}
We thank Matan Gavish for fruitful discussions and for introducing us with Farrel's work \cite{Farrell}.

\end{document}